\documentclass[12pt]{article}
\usepackage[utf8]{luainputenc}
\usepackage{geometry}
\usepackage{array}
\usepackage{rotating}
\usepackage{float}
\usepackage{multirow}
\usepackage{amsmath}
\usepackage{amssymb}
\usepackage{graphicx}
\usepackage{setspace}
\usepackage{esint}
\usepackage{chngcntr}
\usepackage{enumerate}
\usepackage{graphicx}
\usepackage[authoryear]{natbib}
\usepackage{rotating}
\usepackage{changes}
\usepackage{subfig}
\usepackage{url}
\usepackage{xcolor}

\usepackage{amsfonts}\usepackage{amsthm}\usepackage{fullpage}
\usepackage{multirow}\usepackage{enumerate}\restylefloat{table}\usepackage{geometry}
\theoremstyle{plain}
\theoremstyle{definition}
\newtheorem{proposition}{Proposition}

\title{Green antitrust conundrum: Collusion with social goals}

\author{{Nigar Hashimzade} \thanks{Brunel University London, UK, Nigar.Hashimzade@brunel.ac.uk} \and {Limor Hatsor} \thanks{ Jerusalem College of Technology, Israel, limor.hatsor@gmail.com} \and {Artyom Jelnov} \thanks{Ariel University, Israel, artyomj@ariel.ac.il}}

\begin{document}
\maketitle
\begin{abstract}
    Recent antitrust regulations in several countries have granted exemptions for collusion aimed at achieving environmental goals. Firms can apply for exemptions if collusion helps to develop or to implement costly clean technology, particularly in sectors like renewable energy, where capital costs are high and economies of scale are significant. However, if the cost of the green transition is unknown to the competition regulator, firms might exploit the exemption by fixing prices higher than necessary. The regulator faces the decision of whether to permit collusion and whether to commission an investigation of potential price fixing, which incurs costs. We fully characterise the equilibria in this scenario that depend on the regulator’s belief about the high cost of green transition. If the belief is high enough, collusion will be allowed. We also identify conditions under which a regulator’s commitment to always investigate price fixing is preferable to making discretionary decisions.
\end{abstract}
\section{Introduction}
How should governments support green transition and sustainable business practices? In addition to the traditional policy tools, such as various environmental taxes and subsidies, pollution standards and tradeable permits, policy makers in several countries recently turned to the use of antitrust regulation. The green antitrust approach allows exemptions for collusive agreements driven by environmental goals, for example, when collusion enables firms to develop and implement costly clean technologies. An exemption can be granted on certain conditions, typically requiring that it will not be used for price fixing. 

When the regulator cannot observe the actual cost of green transition, the firms can overstate the cost and exploit their increased market power to fix prices at an excessively high level that harms consumers. We analyse this situation in a general inspection game framework, where the decisions on exemption and subsequent investigation are separated. The regulator (the competition authority) decides whether to allow collusion and appoints an inspector, who subsequently decides whether to investigate potential price fixing by the firms. We show that there is a threshold in the regulator’s prior belief about the actual cost of green transition. If the prior is above this threshold, collusion is permitted. Furthermore, when the investigation decision is discretionary (made after observing high-price collusion), the inspector never conducts an investigation if the prior belief about the transition cost is sufficiently high. As a result, firms always violate the exemption provisions whenever possible, fixing high prices even when the actual transition cost is low. However, under medium prior beliefs, the regulator not only allows collusion, but there is also  an incentive to investigate, leading to a mixed-strategy equilibrium where firms violate with a positive probability. We also compare the ex-post discretionary decision to an ex-ante commitment by the inspector to always investigate potential price fixing, identifying the circumstances where the commitment policy is socially preferable. 

In this paper we model a situation where firms seek an exemption from a competition authority to collude, with a claim of implementing a new technology that is “greener” than the old technology currently in use. Developing and implementing the new technology is costly, and the colluding firms will raise the price for their product if it is adopted. The regulator does not observe the transition cost. Therefore, the firms can overstate the cost and fix the price at an excessively high level that harms consumers. If the inspector observes high post-collusion price, it may decide to intervene and investigate the firms’ claim, at a cost of, for example, hiring experts and collecting evidence. At the outset, the regulator chooses whether to allow collusion, based on its belief about the transition cost, and, subsequently, after observing high price, the inspector decides whether to investigate the potentially unjustified price-fixing. In this framework, we show that there is a threshold in the regulator’s belief that the cost is high, above which the collusion is allowed, and a higher threshold above which the inspector never investigates, and therefore firms always misbehave (set high prices even if the transition cost is low). Within these thresholds, we find a mixed-strategy equilibrium, where the incentive of the regulator to investigate enforces firms' good behavior with a positive probability.    

As an alternative to a discretionary ex-post decision about investigation (or a discretion policy), the inspector can commit ex-ante to always investigate high-price collusion. Compared to the discretion policy, the commitment policy is costly, since the inspector pays the investigation cost whenever high-price collusion occurs. Nevertheless, by committing to investigate the inspector prevents violation completely, augmenting  social gains from collusion. Thus, collusion is socially desirable under the commitment policy even in cases of low priors about the transition cost, where it would be  prohibited under the discretion policy. In these circumstances, for a sufficiently small investigation cost, the regulator prefers the commitment policy over the discretion policy. 

The same intuition applies to the case of the high priors, where collusion is allowed under both investigation policies and  the discretion policy yields a \textit{no investigation - always violation} equilibrium. By shifting from the discretion policy to the commitment policy we gain a reduction in the probability of violation from 1 to 0. Thus, when the investigation cost is sufficiently low, the regulator prefers the commitment policy over the discretion policy. 

In contrast, when collusion is allowed under both policies and the discretion policy yields a mixed-strategy equilibrium (intermediate priors), switching to the commitment policy is socially favourable if the investigation cost is \textit{sufficiently high}. To understand the intuition, note that under the discretion policy, the investigation cost negatively affects the incentive to investigate, which, in turn, negatively affects the incentive of the firms to violate. Specifically, an increase in the investigation cost deters the inspector from investigating. Knowing this, the firms are more likely to violate. Therefore, for a sufficiently high investigation cost, the advantage of the commitment policy over the discretion policy in preventing violation is sufficiently large, making the commitment policy socially desirable.  

To summarize, the decision on investigation policy (with or without ex-ante commitment) takes into account two factors: i) the investigation cost (which, in expectation, is larger under the commitment policy than under the discretion policy), ii) the prior belief  about the cost of green technology (which determines whether collusion is allowed, and thereafter affects the investigation decision and, thus, the firms' incentive to violate). Clearly, in case collusion is prohibited under either policy, then social welfare is identical under the two policies. 

Our results suggest the following policy guidelines  when collusion is allowed at least under the commitment policy. Firstly, suppose that the prior about the cost of green technology is close to the \textit{extreme values} (i.e., the transition cost is either most likely to be high or to be low). In these circumstances of high certainty about the transition cost, implement the commitment policy if the investigation cost is \textit{sufficiently low}, else implement the discretionary policy. Secondly, consider the alternative scenario where the prior assumes an \textit{intermediate value} (i.e., the transition cost is about as likely to be high as to be low). In this case, where there is a high level of uncertainty about the transition cost, implement the commitment policy if investigation is \textit{sufficiently costly}, else implement the discretionary policy.

\section{Institutional background and related literature}
Green antitrust trend is a subject of ongoing debate among policy makers, economists, and legal scholars in many jurisdictions. The European Commission (EC) and the competition authorities in several countries, including Austria, Greece, the Netherlands, and the United Kingdom, put forward guidelines allowing to exempt sustainability-driven agreements from the collusion restrictions under competition law (\citealp{hearn2023antitrust}; \citealp{raskovich2022colluding}). Defenders of such exemptions use economies-of-scale arguments to claim that strict enforcement of competition regulation may impede adoption of green practices. They argue that research and development (R\&D) in clean technology can be prohibitively expensive for a single firm, and that there may be a `first-mover disadvantage', whereby the first firm to adopt expensive green practices loses in competition. Therefore, coordination among actors is required to advance positive environmental outcomes (\citealp{loozen2023eu}). 

An example of a sector that may be particularly affected by green antitrust regulation is the renewable energy industry, because of high capital costs and significant economies of scale associated with renewable energy projects (\citealp{WANG2020115287}). These properties create high barriers of entry in the industry, including substantial capital requirements, complex regulatory processes, and the need for specialised expertise, all potentially limiting competition. Furthermore, firms in the renewable energy sector may be involved in various stages of the supply chain, from manufacturing components to developing projects and providing services. Surely, concentration of market power can cause legitimate concerns. Vertical integration can lead to foreclosure of competitors or limit consumer choice. However, R\&D and innovation in the renewable energy sector are central in driving down costs, improving efficiency, and developing new technologies, and their high cost may require cooperation between market players.   Therefore, governments often support the growth of the renewable energy industry as part of broader environmental and energy policy objectives. Green antitrust legislation may provide well-justified exemptions or special considerations to promote cooperation and collaboration among firms in the industry (\citealp{dolmans2023sustainability}).

One notable example of such an exemption is the case of the Energieakkoord (\citealp{dutch2013policy}), the Energy Agreement for Sustainable Growth initiated by the government of the Netherlands. A precursor to the Integrated National Energy and Climate Plan 2021-2030 (\citealp{dutch2021policy}), this comprehensive agreement between various stakeholders, including the government, businesses, and environmental organizations, aims at advancing the transition to renewable energy and improving energy efficiency. The agreement explicitly encourages collaboration among competitors in the energy sector to achieve common environmental goals: significant reductions in CO2 emissions, a substantial increase in energy from renewable sources, such as wind and solar power, and improvements in energy efficiency across various sectors. To facilitate this collaboration, the Dutch government grants certain exemptions from antitrust laws. This allows energy companies to coordinate their efforts in developing and investing in renewable energy projects without the risk of violating competition regulations. The agreement receives support from regulatory authorities, recognizing that cooperation among energy companies is necessary to meet the national and the European Union (EU) environmental targets effectively.

To claim an exemption, the firms must meet several criteria, which, for example, in the \citet{EC2023} guidelines are stated as “efficiency gains”, “indispensability”, “pass-on to consumers”, and “no elimination of competition”.  Accordingly, in the context of environmental benefits, colluding firms can claim exemption, if (i) collusion will lead to development and implementation of a cleaner production process resulting in pollution reduction, (ii) this pollution reduction goal cannot be achieved by any other available means, (iii) consumers will be the ultimate beneficiaries of the proposed collusion as any negative effect from, say, higher prices will be compensated by better quality and cleaner environment, and (iv) there will be residual competition remaining in the market, i.e. the combined market share of the colluding firms will not exceed a specified threshold.  See \citet{malinauskaite2023competition} for a comprehensive account on relevant regulation approaches in the EU. 

The downside, clearly, is that, in the presence of asymmetric information about the cost of adoption of the new technologies and imperfect information about their benefits, it is difficult to ascertain that the criteria are satisfied, allowing firms to take unfair advantage of the relaxation of regulations. The lack of information available to regulators is especially prominent, since the \citet{EC2023} guidelines expand the interpretation of "consumer benefit" from  consumers affected directly in the market to “collective” consumer benefits. In this sense, green antitrust is related to the so-called “hipster anti-trust” thinking, which states that regulations should be assessed on the criteria of wider public interest rather than focusing on the consumers in the affected markets (\citealp{dorsey2018hipster}). While this step aims to take into account environmental externalities, in practice it is notoriously difficult to quantify the efficiency gains accruing to the general population outside the given market, especially since competition authorities are unlikely to have the required expertise (\citealp{raskovich2022colluding}). This expertise belongs to the state agencies tasked with environmental regulations, and handing this task to an agency responsible for consumer protection from market power may potentially result in failure of environmental regulation (\citealp{tirole2023socially}).

The critics of green antitrust further argue that allowing such collusion can lead to “green inflation” harming consumers, whereas competition would incentivise the search for more efficient technologies, ultimately reducing the cost of green transition (\citealp{veljanovski2022case}). Moreover, the emphasis on environmental considerations incentivises “greenwashing”, while the existing legal framework already contains general provisions for exemptions that would apply to “green” collusion cases on the same grounds (\citealp{tirole2023socially}). Thus, for example, an exemption based on a development of a new green technology that requires collusion to meet its costs would fall under the general exemption based on a beneficial but costly R\&D. 

Note that even when the criterion of consumer benefit is restricted to the directly affected market, asymmetric information about the costs hinders verification of the criteria. For example, to assess whether ``pass-on to consumers'' criterion is met the regulator needs to verify whether any price increase by cooperating firms reflects higher cost of green transition, rather than deliberate price-fixing by the firms, and that there is a net increase in consumer welfare. Verification is costly, as it requires collection of information and expert advice, possibly followed by lengthy legal proceedings. An important issue we explore in this article is, therefore, how the public bodies should make a decision on an exemption and on any subsequent inspections of the firms in the green antitrust cases.

Collusion is widely  studied in the economic literature, with seminal contributions by \citet{stigler1964theory} and \citet{green1984noncooperative} (see \citealp{feuerstein2005collusion}, for a review). A key question in much of this literature examines the conditions under which collusion is self-enforced. Another strand of research analyses government intervention to prevent an illegal collusion; see \citet{block1981deterrent}, \citet{cyrenne1999antitrust} and \citet{harrington2004cartel}, among others. A standard assumption in this literature is that collusion detection probability depends on the change of price. In our paper, we assume that collusion can be legal, that is, under certain conditions collusion is exempt from the restrictions on anti-competitive practices. Moreover, once collusion is allowed, the firms commit to it, so there is no issue of self-enforcement. Nevertheless, in case of high-price collusion, the inspector may investigate whether the firms took advantage of the collusion to increase prices excessively (more than necessary for the adoption of the costly green technology), and impose fines ex-post. We refer to this behaviour of the firms as `violation', since this excessive price-fixing violates the conditions of exemption.

Related papers by \citet{schinkel2017can} and \citet{schinkel2022production} study possible collusion with environmental goals as a two-stage decision: investment in sustainability and product competition. They show that collusion in production, but not in investment in green technology, improves social welfare and sustainability. \citet{inderst2023firm} analyse incentives of firms to adopt a green technology when consumers' environmental preferences are shaped by perceived social norms. The authors show that sustainability agreements, or coordination between firms driven by sustainability goals, can be justified under certain conditions, thus providing theoretical support for green antitrust. The novelty of our work lies in examining the possibility of ex-post intervention by the inspector, which influences the firms' decisions to exercise their market power, bolstered by the ex-ante collusion exemption.

This paper is also related to the literature on the inspection games. See \citet{avenhaus2002inspection} for a survey on the decision to perform a costly inspection. Our analysis focuses on an exemption decision while considering whether or not to implement a costly investigation of potential violations by the firms that may follow.

\section{The model \label{sec:model}}
Consider several competing firms that produce a homogeneous good using a traditional, or a \textit{dirty} technology. The dirty technology pollutes the environment but entails no additional cost on the firms. The good can be produced using a new, \textit{clean} or \textit{green} technology. The switch to the green technology is costly, and the transition cost is private information of firms. For simplicity, we assume that the transition cost is exogenously given and can take one of two values, either \textit{low} or \textit{high}. If the transition cost is high, firms can cover it only by colluding on high prices. The antitrust regulations contain a provision allowing firms to collude on prices if it is necessary for the green transition. 

At the first stage in the model, the regulator announces whether collusion between firms is allowed. The type of the technology used by the firms, clean or dirty, is observable. If the regulator allows collusion, then the colluding firms can no longer use the dirty technology: a firm that attempts to use the dirty technology in collusion will be effectively banned from the market, for example, by prohibitively high penalties. 

At the second stage, if collusion is prohibited, firms simultaneously choose the technology, clean or dirty. Production takes place and the game ends. The payoffs depend on the technology choice and the green transition cost in the following way. If the cost of green transition is high, the firms will use the dirty technology. We denote the payoff of firm $i$ from this outcome by $v_{D,i}$, and the corresponding welfare considered by the regulator by $w_D$. The index \textit{D} stands for the adoption of a dirty technology. $w_D$, as well as the regulator's welfare in the following cases to be defined in the sequel, considers both the consumer welfare and the producer welfare (the private welfare).

If the transition cost is low, firms may adopt the  costly clean technology even if collusion is not allowed. The motive for this self-discipline of the firms may derive from consumers' environmental awareness. That is, environmental awareness can boost the demand for clean products to the extent that the firms' profits are higher with the clean technology than with the dirty technology. Denote the payoff of firm $i$ when collusion is not allowed and the green transition cost is low by $v_{G,i}$ and the corresponding welfare considered by the regulator by $w_G$. The index \textit{G} stands for the transition into a green technology. 

If the regulator allows collusion, then at the second stage firms collude and are obliged to pay the transition cost and use the clean technology. While they cannot choose the technology at the second stage, they can choose the price for their output above the competitive level. For simplicity, we assume that the firms can collude either on the \textit{high} or on the \textit{low} price. To cover the transition cost to the green transition when it is high, the firms must collude on the high price. Then, the payoff of firm $i$ is $v_{H,i}$ and the corresponding consumer considered by the regulator is $w_H$. The index \textit{H} stands for high price. In contrast, if the green transition cost is low, setting a low price is sufficient to cover the firms' expenses. If they collude on low price, the payoff of firm $i$ is $v_{L,i}$, and the corresponding welfare considered by the regulator is $w_L$. The index $L$ stands for low price.

However, since the cost is the firms' private information, they may take advantage of the asymmetric information to make even higher profits by colluding on high price when the transition cost is low (violate the conditions of exemption). This behavior yields payoff of $v'_{H,i}$ for firm $i$ and corresponding regulator's welfare of $w^{'}_H$. Excessively high price leads to a loss of consumer surplus that outweighs the profit gains, resulting in a welfare loss,  $w_H^{'} \leq w_L$. Note that to assist the reader with notations, in the Appendix we summarise regulator's welfare in each case.

The regulator observes the price set by the colluding firms. If the price is low, no action is taken, the game ends and the payoffs are realised. If the price is high, the game continues to the third stage, the inspection decision. Next, we analyse the version of the game where the inspector's decision on investigation of potential price-fixing is discretionary. 

\subsection{Discretionary investigation policy}
At the third stage, having observed high-price collusion, the inspector decides whether or not to investigate potential excessive price-fixing by the firms. Without investigation, the inspector pays no cost and receives no benefits. The investigation is costly but it allows establishing with certainty the actual cost of green transition. If the investigation finds no violation (the actual transition cost is high), the inspector obtains no benefit from the investigation but bears the investigation cost of $d\ge 0$. If violation is established (the actual transition cost is low), the firms pay a total fine of $g>d$ to the inspector, equally divided, so that each firm $i$ pays $f$. We assume that the fine imposed on each firm is larger than its gain from violation $f>v'_{H,i}-v_{L,i}$. The inspector carries out an investigation if its expected payoff is positive, namely its expected benefit is larger than the investigation cost. This completes the model with discretionary investigation policy. 

\section{Results \label{sec:results}}
The game is summarized in Figure \ref{fig:game_tree} and is solved by backward induction.

\begin{figure}[H]
    \centering
    \includegraphics[scale=0.65]{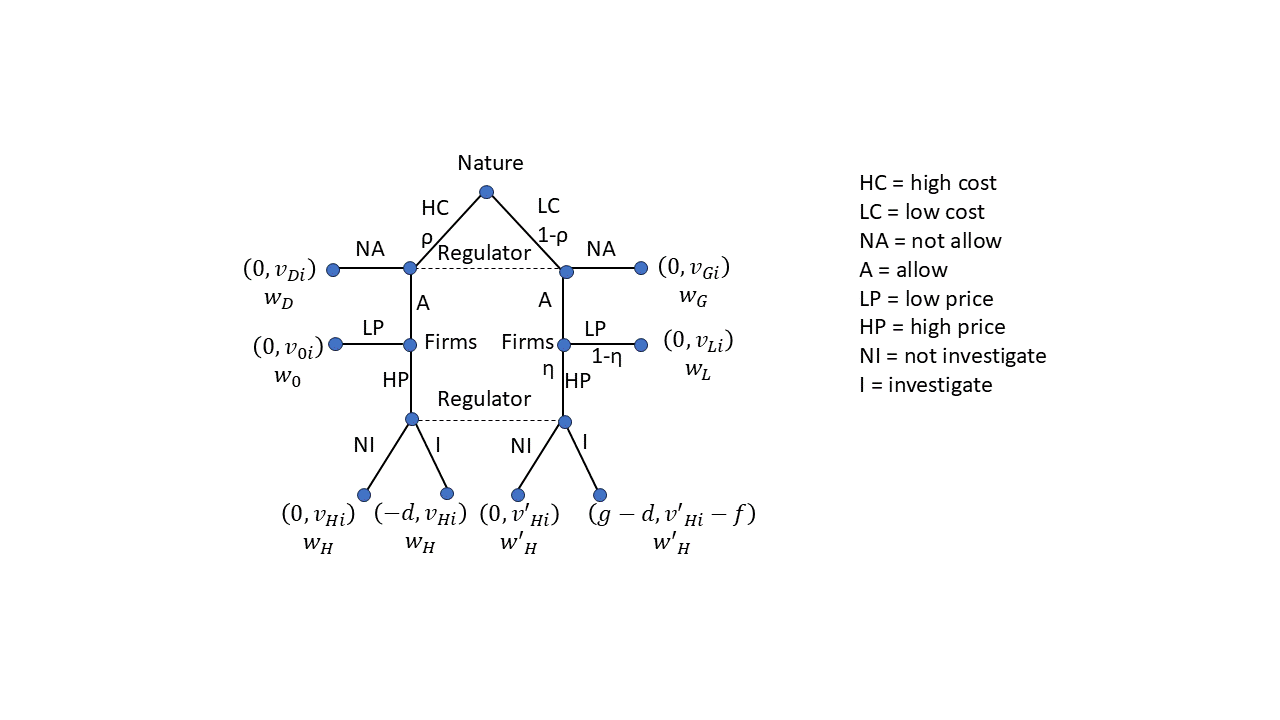}
    \caption{Description of the game with discretionary investigation; $\rho$ is the probability of high transition cost; $\eta$ is the probability of violation by the firms.}
    \label{fig:game_tree}
\end{figure}

 Suppose that at the outset the regulator believes that the transition cost to green technology is high with probability $\rho$ and that, once exempt, the firms will violate the exemption conditions (set high price when the cost is low) with probability $\eta$. Then, the  regulator's expected welfare under collusion is 
\begin{equation}
E\left[W^C\right]=\rho w_H+(1-\rho)[\eta w_H^{'}+(1-\eta) w_L].
    \label{welfare_collusion}
\end{equation}
Note that the regulator's expected welfare when collusion is allowed decreases in $\eta$ because $w_H^{'} \leq w_L$. 

If collusion is prohibited, the regulator's expected welfare is given by 
\begin{equation}
E\left[W^{NC}\right]=\rho w_D+(1-\rho) w_G.
    \label{welfare_nocollusion}
\end{equation}

One can see directly from the regulator's welfare with and without collusion, that if $w_D>w_H$ and $w_G>w_L$, then the regulator blocks collusion. The intuition is the following. When the transition cost is low, environmental awareness of consumers may lead to a clean environment even without collusion. If consumers' demand for goods produced with a clean technology is sufficiently high, $w_G>w_L$, then the switch of the firms to clean technology is self-enforced by the need to accommodate consumers' behaviour. The rise in the revenues due to the high demand outweighs the transition cost, and therefore there is no need for collusion. In this case, the consumers benefit both from buying the good at a competitive price and from the transition to green technology. 

The rationale for allowing collusion emerges when firms are unable to cover the transition cost in a competitive market. Then, when the transition cost is high, allowing collusion is the only means to achieve environmental goals. In these circumstances, a necessary condition for collusion to be socially justified is that the environmental gain is larger than the welfare loss caused by high-price collusion, $w_H>w_D$ (this broadly falls under the criteria of the \citet{EC2023} guidelines). Note also that the lower the price set by the colluding firms, while still allowing them to cover the transition cost, the more socially desirable becomes collusion. 

Accordingly, we assume hereinafter that collusion provides social gains in the case of high transition cost, $w_H>w_D$, but not in the case of low transition cost, $w_G>w_L$, where environmental awareness of consumers would be sufficient to incentivise the transition of firms to clean technology.\footnote{An alternative assumption will provide a less realistic corner solution, where collusion is always allowed or never allowed, as we clarify hereinafter.} Additionally, recall that violation by the firms reduces the regulator's welfare from collusion, $w_L>w_H^{'}$. Then, 

 \begin{equation}
w_D<w_H<w_H^{'}<w_L<w_G
    \label{w_inequality}
 \end{equation}

Therefore, from the gap between the regulator's welfare with and without collusion (subtracting equation \eqref{welfare_nocollusion} from equation \eqref{welfare_collusion}), it is easy to see  that the   expected gain from collusion increases with $\rho$, the probability that the transition cost is high, 

\[E\left[W^C\right]-E\left[W^{NC}\right]=\rho (w_H-w_D)+(1-\rho)[\eta w_H^{'}+(1-\eta) w_L-w_G].\]

By equating \eqref{welfare_nocollusion} and \eqref{welfare_collusion}, we obtain the collusion threshold $\rho^{*}$, above which the regulator allows collusion:
\begin{equation}  \label{threshold}
\rho^{*}(\eta)=\frac{w_G-(\eta w_H^{'}+(1-\eta)w_L)}{w_G-(\eta w_H^{'}+(1-\eta)w_L)+(w_H-w_D)}, 0<\rho^{*}(\eta)<1.
\end{equation}

If $\rho$ is above the threshold $\rho^{*}$ (there is a sufficiently high probability that the transition cost  is high), then allowing collusion increases the regulator's expected welfare, else collusion is prohibited. 
Note that inequality \eqref{w_inequality} assures an interior solution, $0<\rho^{*}(\eta)<1$.\footnote{It is easy to see that $0<\rho^{*}(\eta)<1$, since collusion is socially desirable in the case of high transition cost ($w_H>w_D$) but not in the case of the low transition cost ($w_G>w_L$), and violation by the firms reduces the regulator's welfare from collusion, ($w_L>w_H^{'}$). If $w_H=w_D$, then $\rho^{*}=1$ and collusion is never allowed. Alternatively, if $w_G=\eta w_H^{'}+(1-\eta)w_L$, then $\rho^{*}=0$ and collusion is always allowed.} We obtain the extreme values of collusion thresholds, $\rho^{_L}$ and  $\rho^{_H}$, in case the firms never violate ($\eta=0$) or always violate ($\eta=1$), respectively:
\begin{equation}  \label{threshold pl ph}
\rho^{_L}=\rho^{*}(0)=\frac{w_G-w_L}{w_G-w_L+w_H-w_D}, 
\rho^{_H}=\rho^{*}(1)=\frac{w_G-w^{'}_H}{w_G-w^{'}_H+w_H-w_D}.
 \end{equation}
Clearly, the collusion threshold when firms always violate is higher than in case they never violate $\rho^{_H}>\rho^{_L}$, since $w_L>w^{'}_H$. The exemption provision states that when the transition cost is low, the firms should collude on the low price. If firms follow the purpose of the exemption provision $\eta=0$, the regulator's welfare from collusion is larger than in case they take advantage of the asymmetric information and unjustifiably collude on the high price, $\eta=1$. For intermediate probabilities of violation, $0<\eta<1$, the collusion threshold lies between the extremes, $\rho^{_H}>\rho^{*}(\eta)>\rho^{_L}$, and depends on $\eta$. As the probability of violation, $\eta$, increases, the regulator is less inclined to approve collusion, thereby the collusion threshold increases,
\begin{equation}
\frac{\partial \rho^{*}(\eta)}{\partial \eta}>0,
\frac{\partial^{2} \rho^{*}(\eta)}{\partial^{2} \eta}<0. 
\end{equation}

Alternatively, if violation is effectively prevented, the regulator's gains from collusion rise, inducing the approval of collusion for lower levels of $\rho$. 

This brings us to the analysis of the inspection decision. Our benchmark case is a discretionary investigation policy. Having observed high-price collusion, the inspector decides whether or not to investigate potential excessive price-fixing by the firms. The inspector considers the expected payoff from investigation. The expected benefit from investigation is given by $g\frac{(1-\rho) \eta}{(1-\rho) \eta +\rho}$, where $\frac{(1-\rho) \eta}{(1-\rho) \eta +\rho}$ is the probability of violation (collusion on the high-price when the transition cost is low) given that the observed post-collusion price is high. The inspector chooses to investigate if the expected benefit from inspection exceeds the investigation cost, $d$,

\begin{equation} \label{investigation}
g\frac{(1-\rho) \eta}{(1-\rho) \eta +\rho}>d.
\end{equation}
It is easy to see that this condition always holds if the investigation cost is negligible, $d \to 0$. 

Next, we fully characterise the pure-strategy and mixed-strategy equilibria under the discretionary investigation policy. 

\subsubsection{Pure-strategy equilibria}
We start from mapping the equilibria where the inspector never investigates and consequently the firms choose to always violate, $\eta=1$. Substituting $\eta=1$ in the inspector's condition to investigate \eqref{investigation}, we obtain that the inspector has no incentive to conduct an investigation if $d>g(1-\rho)$ or equivalently $\rho>\frac{g-d}{g}$ ($d>0$). Intuitively, when the clean technology is highly likely to be expensive, collusion on the high price is highly likely to be justified, leaving the inspector with rare potential cases of violation and thus low incentive to investigate. When the inspector never investigates, the firms always violate ($\eta=1$), since they know that they will never be caught and penalised. Note that the fact that firms always violate when they get a chance does not induce the inspector to investigate, because the probability of the transition cost being low, where firms violate, is very small. Thus, the inspector's expected payoff from investigation, $g (1-\rho)-d$, is negative. 

The complete the characterisation of the equilibrium is stated in the following Proposition. 

 \begin{proposition}\label{prop:pure}
 \textit{Collusion allowed -- No investigation -- Always violation equilibrium.}
 
 Let $\rho>\frac{g-d}{g}$ ($d>0$). In the equilibrium, 
  \begin{enumerate}
     
     \item Collusion is allowed if either 
     
     ($i$) $(\frac{g-d}{d})(\frac{w_H-w_D}{w_G-w^{'}_H}) \geq 1$, or 
     
     ($ii$) $(\frac{g-d}{d})(\frac{w_H-w_D}{w_G-w^{'}_H})<1$ and $\rho > \rho^{_H}$. 
     
     Otherwise, collusion is not allowed. 
     
     \item When collusion is allowed, the firms always violate, $\eta=1$,  and the inspector never investigates. 
     
     \item $\frac{\partial \rho^{_H}}{\partial w_G}>0$, $\frac{\partial \rho^{_H}}{\partial w_D}>0$, $\frac{\partial \rho^{_H}}{\partial w_H}<0$, $\frac{\partial \rho^{_H}}{\partial w^{'}_H}<0$.

 \end{enumerate}
 \end{proposition}

\begin{proof}
    Recall that when firms always violate ($\eta=1$), the collusion threshold is given by equation \eqref{threshold}, $\rho^{_H}=\rho^{*}(1)=\frac{w_G-w^{'}_H}{w_G-w^{'}_H+w_H-w_D}$. Then, it is left to map the cases when the inspector never investigates, $\rho>\frac{g-d}{g}$. It is easy to verify that $\rho^{_H}>\frac{g-d}{g} \iff (\frac{g-d}{d})(\frac{w_H-w_D}{w_G-w^{'}_H})<1$. In this case, $\rho^{_H}$ is binding. That is, above $\rho^{_H}$ collusion is allowed and the inspector never investigates. Else, if $\rho^{_H} \leq \frac{g-d}{g}$, then for all $\rho>\frac{g-d}{g} \geq \rho^{_H}$ collusion is allowed and the inspector never investigates. The rest of the proof is straightforward. 
\end{proof}

\subsubsection{Mixed-strategy equilibria}
 Next, we analyse the case of $d<g(1-\rho)$ (or equivalently $\rho<\frac{g-d}{g}$). We first show that in this case there are no equilibria with pure strategies of the firms (always violate, $\eta=1$, or never violate, $\eta=0$). Assume, to the contrary, that firms always violate. Substituting $\eta=1$ in the inspector's condition to investigate \eqref{investigation}, we obtain that the inspector strictly prefers to investigate, because $d<g(1-\rho)$ holds. As the firms know that the inspector always investigates, a violation is fully deterred. However, if the firms never violate, $\eta=0$, then the inspector has no incentive to investigate, because investigation yields a negative payoff. With no violation, the inspector still pays the investigation cost but gets zero benefit. In this case, no investigation encourages the firms to always violate, $\eta=1$, which, again, induces the inspector to investigate. Therefore, when $\rho<\frac{g-d}{g}$, there are no equilibria with pure strategies of the firms. There is, however, an equilibrium with a mixed strategy of violation ($0<\eta<1$). We characterise the mixed strategy equilibrium in the following Proposition. 

\begin{proposition}\label{prop:mixed}
 \textit{Collusion allowed -- Mixed strategy of violation equilibrium.}
 
Let $\rho<\frac{g-d}{g}$. In the equilibrium, 

 \begin{enumerate}
 
    \item Collusion is allowed if $(\frac{g-d}{d})(\frac{w_H-w_D}{w_G-w^{'}_H}) > 1$ and $\rho > \rho^{*}=\frac{(g-d)(w_G-w_L)}{(g-d)(w_G-w_L)+(g-d)(w_H-w_D)-d(w_L-w^{'}_H)}$. Otherwise, collusion is not allowed.
 
     \item When collusion is allowed, the firms violate with probability 
     $\eta=\frac{\rho d}{(1-\rho)(g-d)}$, holding the inspector indifferent between conducting an investigation or not.
       
     \item $\frac{\partial \rho^{*}}{\partial w_G}>0$, $\frac{\partial \rho^{*}}{\partial w_D}>0$, $\frac{\partial \rho^{*}}{\partial w_H}<0$, $\frac{\partial \rho^{*}}{\partial w^{'}_H}<0$, $\frac{\partial \rho^{*}}{\partial w_L}<0$,
     $\frac{\partial \rho^{*}}{\partial(\frac{g-d}{d})}<0$,
     $\frac{\partial \eta}{\partial(\frac{g-d}{d})}<0$, $\frac{\partial \eta}{\partial\rho}>0$.     
      
 \end{enumerate}
\end{proposition}
\begin{proof}
  Firms play a mixed strategy of violation, $0<\eta<1$, if the inspector is indifferent between conducting an investigation or not. Equating the benefit and cost of investigation (equation \eqref{investigation}) yields   
a zero payoff for the inspector $\iff d=g\frac{(1-\rho) \eta}{(1-\rho) \eta +\rho}$.
Rearranging terms, we obtain the probability of violation that satisfies indifference of the inspector,
\begin{equation} \label{eta}
    \eta=\frac{\rho d}{(1-\rho)(g-d).}
\end{equation}
where $\eta<1$ since $\rho<\frac{g-d}{g}$. Substituting  the probability of violation \eqref{eta} in  the collusion threshold \eqref{threshold} yields
\begin{equation}
\rho^{*}=\frac{(g-d)(w_G-w_L)}{(g-d)(w_G-w_L)+(g-d)(w_H-w_D)-d(w_L-w^{'}_H)}. \label{mix_str_threshold}
\end{equation}
It is easy to see that $\rho^{*}\geq \frac{g-d}{g} \iff (\frac{g-d}{d})(\frac{w_H-w_D}{w_G-w^{'}_H})\leq 1$. In this case, for all $\rho<\frac{g-d}{g}\leq\rho^{*}$ collusion is not allowed. Else, $\rho^{*} < \frac{g-d}{g} \iff (\frac{g-d}{d})(\frac{w_H-w_D}{w_G-w^{'}_H})>1$. In this case, $\rho^{*}$ is binding. That is, for all $\rho^{*}<\rho<\frac{g-d}{g}$ collusion is allowed and the inspector is indifferent to investigation. $\rho^{_L}<\rho^{*}<\rho^{_H}$ according to equation \eqref{threshold pl ph}. Namely, $\rho^{*}>\rho^{_L}$, because $w_L>w^{'}_H$, and $\rho^{*}<\rho^{_H} \iff 
d(w_G-w^{'}_H)<(g-d)(w_H-w_D)$. Rearranging terms yields $(\frac{g-d}{d})(\frac{w_H-w_D}{w_G-w^{'}_H})>1$. This condition is also sufficient to assure that $0<\rho^{*}<1$. It is easy to see that 
$\rho^{*}<1 \iff (g-d)(w_H-w_D)>d(w_L-w^{'}_H) \iff (\frac{g-d}{d})(\frac{w_H-w_D}{w_L-w^{'}_H})>1$, which also assures that $\rho^{*}>0$. This condition holds, since $(\frac{g-d}{d})(\frac{w_H-w_D}{w_L-w^{'}_H})>(\frac{g-d}{d})(\frac{w_H-w_D}{w_G-w^{'}_H}) \iff w_G>w_L$. Moreover, $(\frac{g-d}{d})(\frac{w_H-w_D}{w_G-w^{'}_H})>1$ is also a necessary and sufficient condition for $\frac{\partial \rho^{*}}{\partial w_L}<0$. All other derivatives of $\rho^{*}$ are straightforward. 
\end{proof}

To summarise, there are three types of equilibria, depending on the prior probability $\rho$ assigned to the green transition cost being high. First, let the prior belief be sufficiently high, $\rho>\max[\frac{g-d}{g},\rho^{H}]$. In this case, the regulator allows collusion. The reason is that a high chance of the clean technology being expensive makes it more likely that high-price collusion is necessary to finance the adoption of this technology. Moreover, the inspector is not too concerned with violation, since it is highly unlikely that the actual transition cost is low and thereby firms are unlikely to get a chance to violate. Therefore, the inspector never investigates. Clearly,  the firms gain from a high prior belief $\rho$, because they end up in an equilibrium with high-price collusion where they always violate without facing investigation and punishment. In short, this is the best scenario for the firms. However, from the regulator's perspective, this equilibrium may prove as socially inferior ex-post if the prior belief on the transition cost is wrong. In these circumstances, had the regulator known that the actual transition cost is low, then it could have enhanced social welfare from $w_H^{'}$ to $w_G$ by blocking collusion. 

In the second case, $\rho$ is sufficiently low, $\rho<Min[{\rho^{*},\rho^{H}}]$. If the clean technology is highly likely to involve a low transition cost, there is no social justification for collusion. The regulator blocks collusion, believing that it is highly probable that the firms will find adopting clean technology profitable even in a competitive market.
Due to the imperfect information, this equilibrium is somewhat risky for the regulator, whose welfare is maximized ($w_G$) if, indeed, the actual cost is low, but it is minimized ($w_D$) if the actual cost is high.    

Lastly, for intermediate levels of $\rho$, when $(\frac{g-d}{d})(\frac{w_H-w_D}{w_G-w^{'}_H})>1$ and  $\rho^{*}<\rho<\frac{g-d}{g}$, there is a mixed strategy equilibrium, where the probability of high transition cost is high enough for the regulator to allow collusion ($\rho>\rho^{*}$), yet sufficiently low to make investigation worthwhile in case the firms always violate ($\rho<\frac{g-d}{g}$). In other words, the threat to investigate if the firms always violate is credible. This rules out the equilibrium where the firms always violate and the regulator never investigates. Instead, there emerges an equilibrium with a positive probability of violation that holds the inspector indifferent between investigating or not. It is easy to see that the probability of violation, $\eta=\frac{\rho d}{(1-\rho)(g-d)}$, decreases with $\frac{g-d}{d}$. An increase in $\frac{g-d}{d}$ augments the inspector's payoff, so that the inspector strictly prefers to carry out an investigation. The incentive of the inspector to investigate deters violation. Accordingly, the probability of violation, $\eta$, declines, maintaining the indifference condition. Similarly, a higher $\rho$ reduces the chance to observe violation over the indifference threshold, consequently the inspector strictly prefers not to investigate. To maintain the indifference condition, the probability of violation, $\eta$, must increase. 

For the necessary and sufficient conditions of the mixed strategy equilibrium to hold, two essential factors must be sufficiently large. The first factor, $\frac{g-d}{d}$, reflects the inspector's incentive to investigate. When $\rho<\frac{g-d}{g}$ holds, firms that wish to always violate face a credible threat of investigation. Accordingly, the probability of violation, $\eta$, adjusts in the equilibrium to maintain the inspector indifferent about investigation, as noted above. Specifically, an increase in $g$ (the fines for violation) relative to $d$ (the investigation cost) makes investigation attractive to the inspector, and consequently, to maintain indifference, the probability of violation decreases. The decline in the probability of violation increases the regulator's welfare from collusion, reducing the collusion threshold, $\frac{\partial \rho^{*}}{\partial(\frac{g-d}{d})}<0$. 

The second factor, $\frac{w_H-w_D}{w_G-w^{'}_H}$, reflects the extent of social welfare dominance of the collusion regime over the competition regime. The higher this factor, the more socially desirable is collusion to the regulator. 

\section{Commitment policy to always investigate}
So far, in the benchmark case, we assumed that the inspector, appointed by the regulator, is not committed to investigate. The inspector can exercise discretion whether or not to investigate if the collusive (ex-post) price is high. Next, we introduce an alternative policy of commitment to investigate high-price collusion. The commitment policy is  costly, since the inspector always investigates high-price collusion. Nevertheless, this policy always prevents violation ($\eta=0$), making collusion more attractive for the regulator compared to the discretion policy. Therefore, the threshold for the regulator to allow collusion under the commitment policy must be smaller than the threshold in the benchmark case (the discretion policy). Accordingly, it is straightforward that when collusion is not allowed under the commitment policy, it is not allowed also under the discretion policy. To obtain the threshold for collusion in the commitment policy, we substitute $\eta=0$ (no violation) in equation \eqref{threshold pl ph},

\[\rho(0)=\rho^L=\frac{w_G-w_L}{w_G-w_L+w_H-w_D}.\]

Note that in the trivial case where collusion is not allowed in both regimes, $\rho<\rho^L$, there is nothing to violate, and thereby nothing to investigate. Given that investigation never takes place, there is no difference between the regimes. In the sequel, we disregard this trivial case and compare the regimes assuming that collusion is allowed at least under the commitment policy, i.e., $\rho>\rho^L$.

\section{Commitment policy vs. discretion policy}
To compare the two scenarios, with and without commitment, we introduce a social planner, who considers both the regulator's welfare and the inspector's payoff. Accordingly, the social planner's objective, $\pi_{S}$, is defined as a weighted sum of the regulator's welfare and the inspector's payoff:
\[\pi_{S}=\delta_1 w_i+\delta_2 h(gk-d),\]
where $0\leq \delta_1 , \delta_2\leq1$, $\delta_1+\delta_2=1$ are the weights (or intensities) of the two components, measuring their relative importance. This general objective function allows the social planner to completely prefer the inspector if $\delta_1=0, \delta_2=1$, or completely favour the regulator, if $\delta_1=1, \delta_2=0$. 
$w_i \in \{w_L, w_D,w_H,w_H^{'}, w_G\}$ is the realization of regulator's welfare, depending on the decision about collusion, prices and cost of the clean technology; $h$ is $1$ if inspection is performed and $0$ otherwise; $k$ is $1$ in case of violation and $0$ otherwise. Note that under the always investigate policy, $h=1$, there is no violation,  $k=0$, which yields  $\pi_{S}=\delta_1 w_i-\delta_2 d$. Let $E\pi_{S}$ be the expected value of $\pi_{S}$.

\begin{proposition} 
 \textit{Comparison of the commitment policy and the discretion policy}
\label{prop:compare_investigation_policy}   
    
    Let $\rho>\rho^L$. 
    \begin{enumerate}
    \item \textit{High levels of $\rho$}: Suppose that $\rho>\frac{g-d}{g}$ and the benchmark case equilibrium is \textit{Collusion allowed -- No investigation -- Always violation} (see Proposition \ref{prop:pure}).
    Then, the social planner prefers the commitment policy if $d<\frac{\delta_1(1-\rho)(w_L-w^{'}_H)}{\delta_2 \rho}$. Else, the social planner prefers the discretion policy.
    
    \item \textit{Intermediate levels of $\rho$}: Assume that $\rho<\frac{g-d}{g}$ and the benchmark case equilibrium is \textit{Collusion allowed -- Mixed strategy of violation} (see Proposition \ref{prop:mixed}). 
    Then, the social planner prefers the commitment policy if $d>g-\frac{\delta_1 (w_L-w_H^{'})}{\delta_2 }$. Else, the social planner prefers the discretion policy.
    
    \item \textit{Small levels of $\rho$}: Assume that collusion is prohibited in the benchmark case (see Propositions \ref{prop:pure} or \ref{prop:mixed}). 
    Then, the social planner prefers the commitment policy if $d<\frac{\delta_1[\rho(w_H-w_D)+(1-\rho)(w_G-w_L)]}{ \delta_2 \rho}$. Else, the social planner prefers the discretion policy.
    \end{enumerate}
\end{proposition}

\begin{proof}
    We compare the social planner's expected welfare, $E\pi_{S}$, under the commitment policy and under the discretion policy, and obtain conditions under which the social planner prefers the commitment policy. Else, the social planner prefers the discretion policy.
    \begin{enumerate}
       \item In this case, the social planner prefers the commitment policy if
       \[\delta_1[\rho w_H+(1-\rho) w^{'}_H]<\delta_1[\rho w_H+(1-\rho) w_L]-\delta_2 \rho d.\]
       By rearranging terms, we obtain $d<\frac{\delta_1(1-\rho)(w_L-w^{'}_H)}{\delta_2 \rho}$.
        
        \item In this case, the commitment policy is superior to the discretion policy if
        \[\delta_1[\rho w_H+(1-\rho) ((1-\eta)w_L+\eta w_H^{'})]<\delta_1[\rho w_H+(1-\rho) w_L] -\delta_2 \rho d,\]
        which is equivalent to \[\delta_2 \rho d< \delta_1(1-\rho) \eta (w_L-w_H^{'}).\]
        Substituting the probability of violation and rearranging, \eqref{eta}, this inequality holds when
        \[d>g-\frac{\delta_1 (w_L-w_H^{'})}{\delta_2 }.\]
        
        \item In this case, the social planner prefers the commitment policy if
        
        $\delta_1[\rho w_H+(1-\rho) w_L]-\delta_2 \rho d>\delta_1 [\rho w_D+(1-\rho) w_G]. \iff d<\frac{\delta_1[\rho(w_H-w_D)+(1-\rho)(w_G-w_L)]}{\delta_2 \rho}$.   
         \end{enumerate}
\end{proof}

Proposition \ref{prop:compare_investigation_policy} compares  the social planner's expected welfare under the commitment policy with that under the discretion policy. From the conditions on the investigation cost, $d$, in Proposition \ref{prop:compare_investigation_policy}, it is easy to see that when the social planner places sufficiently high importance on the regulator's welfare relative to the inspector's net payoff,  $\frac{\delta_1}{\delta_2}$, the commitment policy is socially preferable. That is, there exists some $d$ which satisfies the inequalities in Proposition \ref{prop:compare_investigation_policy} without violating the conditions of Propositions \ref{prop:pure} or \ref{prop:mixed}. In this case, the social planner is willing to implement the costly commitment policy in order to fully prevent violation.
 
Proposition \ref{prop:compare_investigation_policy} compares the two regimes across three levels of the regulator's prior on the transition cost to green technology (sufficiently high, sufficiently low, and intermediate values of $\rho$). 

1. For extreme priors, there is enough certainty of the regulator that the transition cost to green technology is either high or low. In these circumstances, the commitment policy is socially preferable if the investigation cost $d$ is sufficiently low. This result is quite intuitive. The social planner favours the commitment policy if it is not too costly. 

This result is obtained whenever $\rho$ is sufficiently high, namely the discretion policy provides the \textit{Collusion allowed -- No investigation -- Always violation equilibrium}, or $\rho$ is sufficiently low, namely collusion is not allowed under the discretion policy (but is allowed under the commitment policy, $\rho>\rho^L$).  

2. For the intermediate levels of $\rho$, the transition cost is about as likely to be high as to be low. In other words, there is high uncertainty of the regulator whether the transition cost is high or low. In these circumstances, under the discretion policy we obtain the \textit{Collusion allowed -- Mixed strategy of violation equilibrium}, where firms violate with a positive probability. Then, surprisingly, the social planner prefers the commitment policy over the discretion policy if the investigation cost is sufficiently \textit{high}. Why do the social planner prefers the commitment policy when it becomes sufficiently costly? The reason is that at the same time the social benefit from the commitment policy increases relative to the discretion policy. The social planner prefers the commitment policy when its added expected social benefit (from fully preventing violation) is larger than its added expected cost (because the inspector always investigates high-price collusion). When the investigation cost $d$ increases there are two opposing effects on the social planner's expected welfare. On one hand, higher $d$ makes the commitment policy more costly. On the other hand, in the mixed-strategy equilibrium, the higher the investigation cost $d$, the lower the inspector's incentive  to conduct an investigation, raising the probability of violation under the discretion policy. Consequently, with rising probability of violation, the relative benefit from the discretion policy declines relative to the commitment policy (that prevents violation completely). Proposition \ref{prop:compare_investigation_policy} shows that the latter effect prevails, namely in these circumstances the commitment policy is socially preferable for sufficiently high levels of $d$. 

\section{Concluding remarks}
The green antitrust approach uses competition policies to address environmental issues. This approach is relatively new, and even though it is gaining traction in the EU and in other countries, the economic literature on this topic is scarce. We introduce a general green antitrust framework to analyse these policies, where firms may be exempted from anti-collusion regulation. The transition cost to green technology is not observed by the regulator, but may be discovered via ex-post inspection after observing high-price collusion. 

Further research can compare competition policies with other policy alternatives, such as a broad-based tax to subsidize the cost of the clean technology (see \cite{raskovich2022colluding}). Moreover, a straightforward extension to our green antitrust framework would be a mechanism  whereby the firms return to the state or directly to consumers some part of extra surplus gained from collusion after covering the transition cost. This would allow to compare the welfare consequences of the green antitrust policy combined with a transfer from the firms and tax policies combined with subsidies to the firms, when under both policies the cost of the clean technology is private information of the firms and can be misrepresented.

Additionally, in our model, the firms make a binary choice between a high and a low price. It can be extended to include a range of price levels. The outcome then potentially depends on how `greedy' the firms are in setting the prices under collusion and on the environmental awareness of the consumers. Suppose that the regulator believes that the probability of high  transition cost is  small enough to justify investigation, $\rho<\frac{g-d}{g}$. Then, to guarantee that the regulator approves collusion, the firms' strategy would be to self-restrict their prices, ensuring that the social gains from collusion surpass those in a competitive market. As the firms become less `greedy', the threshold for collusion, $\rho^{*}$, declines and collusion is allowed for lower levels of $\rho$ (recall that  $\rho^{*}$ declines with $w_H$ and $w^{'}_H$ and increases with $w_D$, and $w_G$). This negative relation between the greediness of the firms and the willingness of the regulator to allow collusion provides a mechanism that restrains the prices set by the colluding firms. Specifically, the firms can push prices up to the level where $\rho=\rho^{*}$.

It is important to note that the `greediness' of firms under collusion may increase with  environmental awareness of the consumers. For example, suppose that initially $\rho=\rho^{*}$. However, over time the society's concerns about the environment increase, reducing the tolerance towards dirty technologies. In other words, the regulator's welfare when the firms use a dirty technology, $w_D$  (the case of high transition cost and no collusion), declines, say, because the demand for  goods produced with the dirty technology drops. Consequently, the expected social gains from collusion rise, inducing the regulator to allow collusion for lower levels of $\rho$ ($\rho^{*}$, declines), which in turn encourages a greedier behaviour by the firms. The firms can take advantage of the higher environmental awareness to increase prices under collusion, provided that collusion is still allowed. Therefore, green antitrust approach is more likely to lead to higher prices, the more the consumers care about the environment. This conclusion mirrors the outcome of the environmental taxes: a tax that internalises a negative externality generally leads to an increase in price but also to an increase in social welfare. We argue that higher environmental awareness tends to exacerbate this redistribution caused by green antitrust.

Note also that assuming that the regulator is risk-averse rather than risk-neutral will only slightly affect our results. Within our framework, blocking collusion is a riskier policy from regulator's perspective, as welfare can be either maximized at $w_G$, or minimized at $w_D$. In contrast, allowing collusion guarantees welfare of at least $w_H>w_D$, but with an upper bound of $w_L<w_G$. Therefore, our framework indicates that a risk-averse regulator has a potentially stronger inclination to permit collusion (or set a smaller collusion threshold) compared to a risk-neutral regulator.

An attractive feature of our model is that welfare is defined generally, without specifying in-market or out-of-market effects. It is, therefore, applicable in other contexts, whenever firms claim exemption from competition regulation on the basis of public interest, including sustainability in a wider interpretation, such as fair payments to suppliers or support of local employment. 

 Moreover, our framework can be applied to a variety of industrial policies, beyond the exemption from antitrust laws considered in this paper. The general key premise is the authorities' choice between granting or not granting favourable treatment to a firm subject to the firm meeting certain costly requirements or internalising certain externalities, and the firm's choice between being honest or shirking. The government authority cannot observe the firm's behaviour but can commission a costly inspection and punish violation ex-post. From the authority's perspective, if the firms are honest, granting favourable treatment ex-ante is a superior policy choice than withholding it. However, when firms shirk, the incentive to provide them favorable treatment decreases. Thus, the policy choice to either provide or withhold favorable treatment from firms by the government authority is combined with the inspection decision.
  
 One example of such situation is government incentives to  R\&D, in the form of subsidies or tax reliefs. In the UK, small and medium enterprises (SMEs) can apply for R\&D tax relief, administered and monitored by the UK tax authority, His Majesty's Revenue and Customs (HMRC). According to a recent report (\citealp{TaxPolicyAssoc2023SME}), a significant amount of fraud and error was discovered in the claims made under this scheme. The HMRC's analysis of claims from the 2020/21 tax year revealed that about 25.8 percent of all claims by SMEs were either incorrect or fraudulent. The HMRC began systematically analysing R\&D claims starting from the 2020/21 tax year. This analysis showed that a substantial portion of the claims was either fully disallowed or found to be fraudulent. The scale of the problem was much larger than initially reported, with estimates suggesting that the total cost of fraudulent and mistaken claims could be as high as £10 billion. Consistent with our framework, one can argue that in order to alleviate violation, the HMRC has switched from the discretionary inspection of claims to the commitment policy once the perceived probability of violation became large.
 \section*{Conflict of interest statement}
 We declare no conflict of interests.
 \section*{CRediT roles}
 All authors equally contributed to: conceptualization; formal analysis; methodology; writing – original draft; writing – review and editing. 
\bibliography{collusion}

\begin{thebibliography}{}

\bibitem[\protect\citeauthoryear{Avenhaus, Von~Stengel, and Zamir}{Avenhaus et~al.}{2002}]{avenhaus2002inspection}
Avenhaus, R., B.~Von~Stengel, and S.~Zamir (2002).
\newblock Inspection games.
\newblock {\em Handbook of game theory with economic applications\/}~{\em 3}, 1947--1987.

\bibitem[\protect\citeauthoryear{Block, Nold, and Sidak}{Block et~al.}{1981}]{block1981deterrent}
Block, M.~K., F.~C. Nold, and J.~G. Sidak (1981).
\newblock The deterrent effect of antitrust enforcement.
\newblock {\em Journal of Political Economy\/}~{\em 89\/}(3), 429--445.

\bibitem[\protect\citeauthoryear{Cyrenne}{Cyrenne}{1999}]{cyrenne1999antitrust}
Cyrenne, P. (1999).
\newblock On antitrust enforcement and the deterrence of collusive behaviour.
\newblock {\em Review of Industrial Organization\/}~{\em 14}, 257--272.

\bibitem[\protect\citeauthoryear{Dolmans, Lin, and Hollis}{Dolmans et~al.}{2023}]{dolmans2023sustainability}
Dolmans, M., W.~Lin, and J.~Hollis (2023).
\newblock Sustainability and net zero climate agreements--a transatlantic antitrust perspective.
\newblock {\em Competition Law \& Policy Debate\/}~{\em 8\/}(2), 63--80.

\bibitem[\protect\citeauthoryear{Dorsey, Rybnicek, and Wright}{Dorsey et~al.}{2018}]{dorsey2018hipster}
Dorsey, E., J.~Rybnicek, and J.~D. Wright (2018).
\newblock Hipster antitrust meets public choice economics: the consumer welfare standard, rule of law, and rent-seeking.
\newblock {\em Competition Policy International Antitrust Chronicle (April 2018), George Mason Law \& Economics Research Paper\/}~(18-20).

\bibitem[\protect\citeauthoryear{EC}{EC}{2023}]{EC2023}
EC (2023).
\newblock Guidelines on the applicability of article 101 of the treaty on the functioning of the european union to horizontal co-operation agreements.

\bibitem[\protect\citeauthoryear{Feuerstein}{Feuerstein}{2005}]{feuerstein2005collusion}
Feuerstein, S. (2005).
\newblock Collusion in industrial economics—a survey.
\newblock {\em Journal of Industry, Competition and Trade\/}~{\em 5}, 163--198.

\bibitem[\protect\citeauthoryear{Green and Porter}{Green and Porter}{1984}]{green1984noncooperative}
Green, E.~J. and R.~H. Porter (1984).
\newblock Noncooperative collusion under imperfect price information.
\newblock {\em Econometrica: Journal of the Econometric Society\/}, 87--100.

\bibitem[\protect\citeauthoryear{Harrington}{Harrington}{2004}]{harrington2004cartel}
Harrington, J.~E. (2004).
\newblock Cartel pricing dynamics in the presence of an antitrust authority.
\newblock {\em RAND Journal of Economics\/}, 651--673.

\bibitem[\protect\citeauthoryear{Hearn, Hanawalt, and Sachs}{Hearn et~al.}{2023}]{hearn2023antitrust}
Hearn, D., C.~Hanawalt, and L.~Sachs (2023).
\newblock Antitrust and sustainability: A landscape analysis.
\newblock {\em Available at SSRN 4547522\/}.

\bibitem[\protect\citeauthoryear{Inderst, Sartzetakis, and Xepapadeas}{Inderst et~al.}{2023}]{inderst2023firm}
Inderst, R., E.~S. Sartzetakis, and A.~Xepapadeas (2023).
\newblock Firm competition and cooperation with norm-based preferences for sustainability.
\newblock {\em The Journal of Industrial Economics\/}~{\em 71\/}(4), 1038--1071.

\bibitem[\protect\citeauthoryear{Loozen}{Loozen}{2023}]{loozen2023eu}
Loozen, E. (2023).
\newblock Eu antitrust in support of the green deal. why better is not good enough.
\newblock {\em Journal of Antitrust Enforcement\/}, jnad005.

\bibitem[\protect\citeauthoryear{Malinauskaite and Erdem}{Malinauskaite and Erdem}{2023}]{malinauskaite2023competition}
Malinauskaite, J. and F.~B. Erdem (2023).
\newblock Competition law and sustainability in the eu: Modelling the perspectives of national competition authorities.
\newblock {\em JCMS: Journal of Common Market Studies\/}.

\bibitem[\protect\citeauthoryear{Ministry of Economic Affairs}{Ministry of Economic Affairs}{2013}]{dutch2013policy}
Ministry of Economic Affairs (2013).
\newblock {\em Energy Agreement for Sustainable Growth}.
\newblock Government of the Netherlands: Ministry of Economic Affairs.

\bibitem[\protect\citeauthoryear{Ministry of Economic Affairs and Climate Policy (EZK)}{Ministry of Economic Affairs and Climate Policy (EZK)}{2019}]{dutch2021policy}
Ministry of Economic Affairs and Climate Policy (EZK) (2019).
\newblock {\em Integrated National Energy and Climate Plan 2021-2030}.
\newblock Government of the Netherlands, PO Box 20401, 2500 EK The Hague: Ministry of Economic Affairs and Climate Policy (EZK).

\bibitem[\protect\citeauthoryear{Neidle}{Neidle}{2023}]{TaxPolicyAssoc2023SME}
Neidle, D. (2023).
\newblock {\em The £10bn R\&D tax relief scandal: the evidence and the blame}.
\newblock 118 Pall Mall, London, England, SW1Y 5EA: Tax Policy Associates, Ltd.

\bibitem[\protect\citeauthoryear{Raskovich, Ginsburg, Kobayashi, Lipsky, Wright, and Yun}{Raskovich et~al.}{2022}]{raskovich2022colluding}
Raskovich, A., D.~H. Ginsburg, B.~H. Kobayashi, A.~B. Lipsky, J.~D. Wright, and J.~M. Yun (2022).
\newblock Colluding to go green: Global antitrust institute comments on the austrian federal competition authority’s draft guidelines to exempt “sustainability agreements”.
\newblock {\em George Mason Law \& Economics Research Paper\/}~(22-29).

\bibitem[\protect\citeauthoryear{Schinkel and Spiegel}{Schinkel and Spiegel}{2017}]{schinkel2017can}
Schinkel, M.~P. and Y.~Spiegel (2017).
\newblock Can collusion promote sustainable consumption and production?
\newblock {\em International Journal of Industrial Organization\/}~{\em 53}, 371--398.

\bibitem[\protect\citeauthoryear{Schinkel, Spiegel, and Treuren}{Schinkel et~al.}{2022}]{schinkel2022production}
Schinkel, M.~P., Y.~Spiegel, and L.~Treuren (2022).
\newblock Production agreements, sustainability investments, and consumer welfare.
\newblock {\em Economics Letters\/}~{\em 216}, 110564.

\bibitem[\protect\citeauthoryear{Stigler}{Stigler}{1964}]{stigler1964theory}
Stigler, G.~J. (1964).
\newblock A theory of oligopoly.
\newblock {\em Journal of political Economy\/}~{\em 72\/}(1), 44--61.

\bibitem[\protect\citeauthoryear{Tirole}{Tirole}{2023}]{tirole2023socially}
Tirole, J. (2023).
\newblock Socially responsible agencies.
\newblock {\em Competition Law \& Policy Debate\/}~{\em 7\/}(4), 171--177.

\bibitem[\protect\citeauthoryear{Veljanovski}{Veljanovski}{2022}]{veljanovski2022case}
Veljanovski, C. (2022).
\newblock The case against green antitrust.
\newblock {\em European Competition Journal\/}~{\em 18\/}(3), 501--513.

\bibitem[\protect\citeauthoryear{Wang, Hsu, Zheng, Chen, and Li}{Wang et~al.}{2020}]{WANG2020115287}
Wang, R., S.-C. Hsu, S.~Zheng, J.-H. Chen, and X.~I. Li (2020).
\newblock Renewable energy microgrids: Economic evaluation and decision making for government policies to contribute to affordable and clean energy.
\newblock {\em Applied Energy\/}~{\em 274}, 115287.

\end{thebibliography}
\bibliographystyle{chicago}

\section{Appendix (printed): Summary of welfare notations}
In this section, to assist the readers with notations, we summarise regulator's welfare given the realization of the transition cost to green technology, the collusion decision of the regulator and the decisions of the firms about technology adoption and their choice of price. The full description of the regulator's welfare in each case appears in section \ref{sec:model}. 

When collusion is forbidden, the regulator's welfare is either maximized at $w_G$ if the transition cost is low, or minimized at $w_D$ if the transition cost is high. The index \textit{G} stands for the transition into a green technology provided that the transition cost is low, whereas the index \textit{D} stands for the adoption of a dirty technology in case the transition cost is high.

When collusion is allowed, the firms are obliged to implement the green technology, and the regulator's welfare is between the values of $w_G$ and $w_D$. In case the realization of the transition cost is high, firms collude on the high price and the regulator's welfare is denoted by $w_H$, where the index \textit{H} stands for high price. Alternatively, if the actual transition cost is low, the regulator's welfare is larger than $w_H$ and depends on whether the firms violate the exemption provision or not. If the firms violate (choose the high price), the regulator's welfare $w^{'}_H$ is smaller than the regulator's welfare in case they decide to well-behave (choose the low price), $w_L$, where the index $L$ stands for low price. 

\end{document}